\documentclass{article}
\usepackage{graphicx} % Required for inserting images
\usepackage{xcolor}
\usepackage{pgf-umlsd} % sequence diagram
\usepackage{amssymb}
\usepackage{amsmath} % xleftarrow
\usepackage{enumitem}
\setlist[itemize]{noitemsep} % compact itemize
\setlist[enumerate]{noitemsep}
\usepackage{url}

\newtheorem{definition}{Definition}[section]
\newtheorem{theorem}{Theorem}[section]
\newtheorem{lemma}{Lemma}[section]

\newenvironment{proof}{\textsf{Proof}.}{\hfill$\Box$}

\newcommand{\sig}[0]{\mathsf{S}}
\newcommand{\sigver}[0]{\mathsf{V}}

\newcommand{\pubkey}[0]{\mathsf{pk}}
\newcommand{\prikey}[0]{\mathsf{sk}}

\newcommand{\setup}[0]{\mathsf{Set}}
\newcommand{\commit}[0]{\mathsf{Com}}
\newcommand{\open}[0]{\mathsf{Open}}
\newcommand{\commitc}[0]{\mathsf{Com}^{c}}

\newcommand{\param}[0]{\mathsf{par}}

\newcommand{\unisrv}[0]{\mathsf{US}}

\newcommand{\sthash}[0]{h_\mathsf{st}}
\newcommand{\txhash}[0]{h_\mathsf{tx}}

\newcommand{\auxd}[0]{\mathsf{aux}}

\newcommand{\univer}[0]{\mathcal{V}}

\newcommand{\certver}[0]{\mathcal{V}_\mathsf{cert}}

\newcommand{\pinc}[0]{\pi_{\mathsf{inc}}}

\newcommand{\predi}[0]{\nu}
\newcommand{\systime}[0]{\tau}
\newcommand{\newtime}[0]{\mathsf{NT}}
\newcommand{\exttime}[0]{\mathsf{time}}
\newcommand{\prgen}[0]{\mathsf{G}_\mathsf{pr}}
\newcommand{\prsig}[0]{\mathsf{S}_\mathsf{pr}}

\title{ Unicity: Predicates and Atomic Swaps}
\author{
Ahto Buldas$^{1}$ \and
Dirk Draheim$^{2}$ \and
Mike Gault$^{3}$ \and
Risto Laanoja$^{3}$ \and
Vladimir Rogojin$^{3}$ \and
Ahto Truu$^{3}$
\\[1em]
$^{1}$ Tallinn University of Technology, Estonia, ahto.buldas@taltech.ee\\
$^{2}$ Tallinn University of Technology, Estonia, dirk.draheim@taltech.ee\\
$^{3}$ Unicity Labs OÜ, Estonia, ahto.truu@unicity-labs.com
}
\date{\today}

\begin{document}

\maketitle

\begin{abstract}
We generalize Unicity token ownership to programmable spending conditions called \emph{predicates}, enabling smart-contract like functionality executed off-chain directly by relying parties rather than by consensus participants. We prove that the security properties of the Unicity execution layer are preserved under reduction to predicate family unforgeability. To demonstrate the utility of the model, we show how to implement trustless atomic swaps by using predicates.
\end{abstract}

\section{Introduction}

Predicates generalize the concepts of token ownership and transfer in the Unicity infrastructure, which in
the paper \cite{BTLR25}
were defined via digital signatures as follows:
\begin{itemize}
\item \emph{Owner} -- a (legal/physical) person who controls the private key of a digital signature scheme
\item \emph{Ownership condition} -- the public key $\pubkey$ that corresponds to the private key of the owner.
\item \emph{Transfer} -- the owner presents a digital signature $\sigma$ on $m=H(\sthash,\txhash)$ (where $\sthash,\txhash$ are the state hash and the transaction hash, respectively) such that $\sigver(\pubkey,m,\sigma)=1$.
\end{itemize}
In this paper, we present the following generalized concepts:
\begin{itemize}
\item \emph{Owner} -- an abstract group of (legal/physical) persons that together control the information necessary to create the next transaction with the token. The information may include private keys.
\item \emph{Ownership condition} -- a logical condition (predicate) $\predi$.
\item \emph{Transfer} -- the owner (as a group) presents a bit-string $u$ such that the condition $\predi(\systime, m,u)=1$ holds, where
$m=H(\sthash,\txhash)$ and $\systime$ is the \emph{system time} (an integer defined by the Unicity service). This means that predicates may also put restrictions on transaction execution time.
\end{itemize}
So far, we only have used predicates of type $\predi(\tau,m,u)\equiv \sigver(\pubkey,m,\sigma)$, i.e.
all predicates $\predi$ are in the form $\sigver(\pubkey,\cdot,\cdot)$ and .

The predicates approach is certainly not new and is used already in the Bitcoin blockchain, where $\predi$ is called the \emph{locking script} and $u$ is called the \emph{unlocking script} or \emph{witness}.

A typical example of a generalized predicate is the \emph{delayed execution} predicate defined by $\predi(\systime,m,u)\equiv \mathsf{dex}_{\pubkey,\tau_0}(\systime,m,u) \equiv (\systime\ge\tau_0)\;\wedge\; \sigver(\pubkey,m,u)=1$, which states that the next transaction can be executed by the owner of the private key of $\pubkey$ not earlier than $\tau_0$. The delayed execution predicate is used in the protocols for inter-blockchain \emph{atomic swaps} between Bitcoin type blockchains.

It is a natural question whether the security properties (no double spending, no blocking, no association) will still hold in the Unicity infrastructure if the generalized predicates are in use. We will show shortly that double-spending is indeed impossible in the generalized scheme. However, the non-blocking condition is much less obvious.

The main concern is that arbitrarily chosen predicates do not have guaranteed security properties like the UF-CMA condition for digital signatures. For example, if a user chooses the ownership predicate $\predi(\systime,m,u)\equiv (u^2 - 2u + 1=0)$ then anyone who is able to solve quadratic equations can make the next transfer with the token and hence, the no-blocking condition may seem to be violated. On the other hand, by intentionally choosing such an ownership condition, the previous owner may indicate that the next abstract owner of the token is the group of all people who can efficiently solve quadratic equations, and in this sense, intuitively, the no blocking condition is not violated.

The predicates can also be chosen so that they cannot be solved (satisfied) in principle, i.e. they are logically inconsistent. For example
$\predi(\systime,m,u)\equiv (u^2 + 1=0)$, where $u$
is required to be real number, cannot be satisfied.

Another (less trivial) example of improper use of predicates is when a user applies one-time signature scheme as a many-times signature scheme. In this case, security-critical information leaks gradually so that an adversary, having triples
$(\systime_1,m_1,u_1),\ldots,(\systime_n,m_n,u_n)$ (so that $\predi(\systime_i,m_i,u_i)=1$), can construct a new triple $(\systime,m,u)$ such that $\predi(\systime,m,u)=1$ and $m\not\in\{m_1,\ldots,m_n\}$.

Therefore, the precise mathematical formulation of the no-blocking condition -- \emph{only the owner of the private key of $\pubkey$ can block the state $S=(\pubkey,\sthash)$} -- must be revisited. We will redefine the no-blocking condition as follows -- \emph{only those who can solve the predicate $\predi$ can block the state $S=(\predi,\sthash)$}, where by solving $\predi$ we mean finding, for a given $m$, a pair $(\systime,u)$ so that $\predi(\systime,m,u)=1$.
We will make this security definition precise and prove that it holds in the Unicity infrastructure. Intuitively, this means that undesired blocking can happen only because of the weakness of the user-chosen predicates and never because of the structural weakness of the Unicity infrastructure itself.

The paper is organized as follows. In Section~\ref{sec:predicates} we define the predicates and give some examples. In Section~\ref{sec:infra},
we show how to extend the Unicity Infrastructure with predicates.
In Section~\ref{sec:security}, we prove that using predicates will not affect the security properties of the Unicity infrastructure. In Section~\ref{sec:swap}, we discuss how to implement trustless atomic swaps by using predicates.

\section{Predicates}\label{sec:predicates}

In mathematics, a \emph{predicate} $\nu$ is a function
$\nu\colon \mathcal{D}\rightarrow \{0,1\}$ that for every \emph{argument} $d\in \mathcal{D}$ of the
\emph{domain} $\mathcal{D}$ defines a binary value $\nu(d)\in\{0,1\}$.
A \emph{parametrized family} of predicates is a set $\{\nu(\param;\cdot)\}_{\param\in\mathcal{P}}$ such that
for every \emph{parameter} $\param\in\mathcal{P}$, we have a predicate
$\nu(\param;\cdot)\colon \mathcal{D}_\param\rightarrow\{0,1\}$, i.e.
\[
\nu(\param;d)\in\{0,1\}
\]
for every argument $d\in\mathcal{D}_\param$.
In this section, we describe how predicates are described in the Unicity infrastructure and give some examples of predicates.

\subsection{Predicates in Unicity}
Unicity infrastructure needs a somewhat restricted form of predicates.
Every predicate $\nu_\mathsf{name}(\mathsf{par};\cdot)$ in Unicity has:
\begin{itemize}
\item \emph{Name} of the predicate denoted by $\mathsf{name}$
\item \emph{Parameters} of the predicate denoted by $\mathsf{par}$ the type of which depends on $\mathsf{name}$
\item \emph{Domain} $\mathcal{D}=\mathcal{D}_0\times\mathcal{D}_\mathsf{name}$, where $\mathcal{D}_0=\mathbb{T}\times\mathbb{H}$ is the standard part and
$\mathcal{D}_\mathsf{name}$ is the name-dependent part. Here, $\mathbb{T}$ is the set of all possible time values and $\mathbb{H}$ is the range of the hash function $H$ used in the Unicity system.
\end{itemize}
\medskip

\noindent \emph{Arguments} of the predicates are in the form $(\tau,m,u)$, where:
\begin{itemize}
\item $\tau\in\mathbb{T}$ is the system time (explained below)
\item $m\in \mathbb{H}$ is a hash value
\item $u\in\mathcal{D}_\mathsf{name}$ is the unlocking argument, the structure of which depends on both $\mathsf{name}$ and $\mathsf{par}$.
\end{itemize}
\medskip

\noindent\emph{System time} $\tau$ is a non-negative integer held by $\unisrv$ and incremented when $R$ is certified by the BFT layer.
System time can be extracted from inclusion proofs by $\tau'\gets \mathsf{time}(\pi)$, where $\tau'$ is the value of system time when $\pi$ was created.
\medskip

\noindent In a more technical description, the domain also contains a \emph{system information} component $\mathcal{T}$ that enables to verify the inclusion proofs. In actual implementation of the system, $\mathcal{T}$ contains the blockchain verification information that comes from the consensus (BFT) layer of the system, which consists of the block headers of the blockchain. We assume that $\unisrv$ and all users of the system have up to date version of $\mathcal{T}$ and, in this paper, we do not describe or study how this is guaranteed. Therefore, we omit $\mathcal{T}$ for simplicity. \medskip

\noindent
In software implementations, predicates are encoded similarly to public keys -- a predicate name code for $\mathsf{name}$ (like the algorithm identifier of generic public keys) followed by the binary representation of parameters $\mathsf{par}$.

\subsection{Examples of Predicates}
In this section, we provide some examples of predicates that are useful to define in the Unicity infrastructure.

\subsubsection{Signature Predicate}

The signature predicate  $\nu_\mathsf{sig}(\pubkey;\cdot)$ with a single public key $\pubkey$ as the parameter is equivalent to verification of a signatures with the public key $\pubkey$, or in Bitcoin terms, pay to the public key.
\medskip

\noindent\textbf{Standard arguments}: $\tau,m$, i.e. system time $\tau$ and hash of a transaction $m$. As all other predicates have the same standard arguments, we will omit them for the next predicates.\medskip

\noindent\textbf{Name-specific arguments}: $u=\sigma$, i.e. $u$ contains just a single digital signature $\sigma$.\medskip

\noindent\textbf{Definition}: $\nu_\mathsf{sig}(\pubkey;\tau,m,u)=1$ iff $\sigver(\pubkey,m,\sigma)=1$

\subsubsection{P2PKH Predicate}

The Pay-to-Public-Key-Hash (P2PKH) predicate $\nu_\mathsf{p2pkh}(h_\mathsf{pk};\cdot)$ with a public key hash $h_\mathsf{pk}$ as the parameter is equivalent to verification of a signature with a public key that hashes to $h_\mathsf{pk}$.
\medskip

\noindent\textbf{Name-specific arguments}: $u=(\pubkey, \sigma)$, i.e. $u$ contains the public key $\pubkey$ and a digital signature $\sigma$.\medskip

\noindent\textbf{Definition}: $\nu_\mathsf{p2pkh}(h_\mathsf{pk};\tau,m,u)=1$ iff:
\begin{itemize}
\item $H(\pubkey)=h_\mathsf{pk} \mathrel{\land} \sigver(\pubkey,m,\sigma)=1$
\end{itemize}

\subsubsection{Multi-Signature Predicate}
The multi-signature predicate
$\nu_\mathsf{msig}(\pubkey_1,\ldots,\pubkey_n;\cdot)$ with a list $\pubkey_1,\ldots,\pubkey_n$ of $n$ public keys as parameters represents verification of $n$ signatures on the same message hash. \medskip

\noindent\textbf{Name-specific arguments}: $u=(\sigma_1,\ldots,\sigma_n)$, where $\sigma_i$ are digital signatures.\medskip

\noindent\textbf{Definition}: $\nu_\mathsf{msig}(\pubkey_1,\ldots,\pubkey_n;\tau,m,u)=1$ iff:
\begin{itemize}
\item $\sigver(\pubkey_i,m,\sigma_i)=1$ for every $i\in\{1,\ldots,n\}$
\end{itemize}

\subsubsection{Threshold-Signature Predicate}
The threshold signature predicate
$\nu_\mathsf{tsig}(k,\pubkey_1,\ldots,\pubkey_n;\cdot)$ with signature threshold $k\in\{1,\ldots,n\}$ and a list $\pubkey_1,\ldots,\pubkey_n$ of $n$ public keys as parameters is equivalent to verifying $k$ digital signatures with different public keys (in the list) on the same message hash $m$. \medskip

\noindent\textbf{Name-specific arguments}: $u=((\sigma_1,j_1),\ldots,(\sigma_k,j_k))$ contain a list of pairs $(\sigma_i,j_i)$, where $j_i\in\{1,\ldots,n\}$ is the index of the public key that is supposed to be used for verifying $\sigma_i$. \medskip

\noindent\textbf{Definition}: $\nu_\mathsf{tsig}(k,\pubkey_1,\ldots,\pubkey_n;\tau,m,u)=1$ iff:
\begin{itemize}
\item[1.] $\sigver(\pubkey_{j_i},m,\sigma_i)=1$ for every $i\in\{1,\ldots,k\}$
\item[2.] All $j_i$ are different, i.e. $j_i=j_{i'}$ implies $i=i'$
\end{itemize}
The last requirement is necessary because otherwise the predicate can be satisfied by presenting $k$ copies of the same digital signature that verifies with the same public key, say $\pubkey_1$.

\subsubsection{Hashed Timelock Predicate}
The hashed timelock predicate
$\nu_\mathsf{htlc}(\pubkey,\pubkey',y,\tau_\mathsf{max};\cdot)$ with two public keys $\pubkey,\pubkey'$, a hash value $y\in\mathbb{H}$, and a timeout $\tau_\mathsf{max}\in\mathbb{T}$ as parameters is useful for atomic swaps, potentially between different blockchains \cite{Herl18,Bitc18a,Bitc18b}. We describe this predicate just as an example because the swap protocol described in this paper uses different predicates. \medskip

\noindent\textbf{Name-specific arguments}: $u=(x,\sigma)$, where $x$ is a hash value (a pre-image of $y$) and $\sigma$ is a a digital signature.\medskip

\noindent\textbf{Definition}: $\nu_\mathsf{htlc}(\pubkey,\pubkey',y,\tau_\mathsf{max};\tau,m,u)=1$ iff  (a) or (b), where:
\begin{itemize}
\item[(a)] $\sigver(\pubkey',m,\sigma)=1$, $y=H(x)$, and $\tau\le \tau_\mathsf{max}$
\item[(b)] $\sigver(\pubkey,m,\sigma)=1$ and $\tau> \tau_\mathsf{max}$
\end{itemize}
\medskip

\noindent
The owner of $\pubkey'$ is able to satisfy the predicate before system time $\tau_\mathsf{max}$ if he knows the $H$-preimage $x$ of $y$.
The owner of $\pubkey$ is able to satisfy the predicate after $\tau_\mathsf{max}$.

\section{Unicity Infrastructure with Predicates}\label{sec:infra}

\subsection{Unicity Service with Predicates}

Unicity service $\unisrv$ maintains a key/value store $R$ and system time $\tau$, which is a non-negative integer. Initially, $R[k]=\bot$ for every key $k$, and $\tau=0$.\medskip

\noindent Unicity service processes \emph{requests} $Q=(\nu,h,v,u)$,  \emph{status queries} $Q_\mathsf{st}=(k)$, and \emph{new time} messages $NT=(\tau_\mathsf{new})$:
\begin{itemize}
\item $\pi\gets\unisrv(Q)$: For $k=H(\nu,h)$ and $m=H(h,v)$:
\begin{itemize}
\item \textbf{If} $R[k]=\bot$ and $\nu(\tau, m, u)=1$ \textbf{then}:
   \begin{itemize}
      \item Assign $R[k]\gets (v,u)$ and \textbf{return} inclusion proof $\pi$
   \end{itemize}
\item \textbf{Else} \textbf{return} error
\end{itemize}
\item $(V,\pi)\gets \unisrv(Q_\mathsf{st})$: Create inclusion proof $\pi$ for $V=R[k]$ and return $(V,\pi)$. Note that $V=(v,u)$ or $V=\bot$
\item $\unisrv(NT)$: \textbf{If} $\tau<\tau_\mathsf{new}$ \textbf{then} $\tau\gets \tau_\mathsf{new}$.
\end{itemize}
Sometimes we use a shorthand notation $R[k]=v$ instead of $R[k]=(v,u)$\medskip

%\noindent\textbf{Verification function} $\mathcal{V}(k,v,\pi)=1$ iff $R[k]=v$ at $\mathsf{time}(\pi)$
%\medskip

\noindent During the normal work of the system, $\newtime$ is initiated by (and can only executed by) the consensus layer of the system. In attack scenarios, we will also give the adversary the access to the $\newtime$ functionality. This is for making the security conditions stronger. \medskip

\noindent It is easy to see that if $(R_0,\systime_0) = (\emptyset,0)$ is the initial state, $I_1, I_2, \ldots, I_n$ is any sequence of inputs (requests or new time inputs), $R_i$ is the dictionary after the input $I_i$, and $\systime_i$ is the system time after the input $I_i$, then:
\begin{itemize}
\item $\systime_0 \le \systime_1 \le \ldots \le \systime_n$, i.e.
the system time never decreases.
\item $R_0 \subseteq R_1 \subseteq \ldots \subseteq R_n$, i.e. the elements are never removed.
\item The state $R_n$ is a \emph{partial function}, i.e.
$\{(k, v), (k, v')\} \subseteq R$ implies $v = v'$.
\end{itemize}

We say that a key $k$ is \emph{blocked} if $R[k] \neq \bot$.

\subsection{Verification Function}

We assume that the inclusion proofs $\pinc$ contain the system time $\systime$, i.e. the value of $\systime$ when the corresponding new element was set in $R$. There is an extraction function $\exttime$ that
extracts the system time from the proof $\pinc$, i.e. $\systime\gets \exttime(\pinc)$.
We assume that the inclusion proofs $\pinc$ can be verified by any party using a verification function $\univer$ so that:
\begin{itemize}
\item If $\univer(k, v; \pinc) = 1$ then $R[k] = v$ in the current state $(R,\systime)$ of $\unisrv$. Hence, as $R$ is a partial function, for every $k, v, v', \pinc, \pinc'$, the following implication holds:
\begin{equation}\label{eq:eqtx-pred}
\univer(k, v; \pinc) = \univer(k, v'; \pinc') = 1 \quad \Rightarrow \quad v =v' \enspace.
\end{equation}
\item If $R[H(\predi, \sthash)] = \txhash$ after a request $\pinc \gets \unisrv(\predi, \sthash, \txhash, u)$ to the Unicity Service, then $\univer(H(\predi, \sthash), \txhash, \pinc) = 1$.
\item If a request $\pinc\gets\unisrv(Q)$ was processed in the state $S=(R,\systime)$ that changes $R[H(\predi, \sthash)]$ from $\bot$ to $\neq\bot$, then $\exttime(\pinc)=\systime$.
\end{itemize}

\subsection{Transactions with a Token}

Every token has a \emph{state hash} $\sthash$ and an abstract owner $A$ represented by a predicate $\predi$. The state hash $\sthash$ is initialized by the \emph{mint transaction} of the token. We will call the pair $(\predi, \sthash)$ the \emph{state} of the token.

The transaction payload (before certifying the transaction) with the token is a pair
$T = (\sthash, D)$, where:
\begin{enumerate}
\item $\sthash$ is the state hash before executing the transaction,
\item $D$ (transaction data) contains the following fields:
\begin{itemize}
\item $\predi'$: the predicate of the next abstract owner,
\item $x$: a uniformly chosen random string $x\gets\{0,1\}^\ell$,
\item $\auxd'$: auxiliary data for the next state.
\end{itemize}
\end{enumerate}

The pair $(\predi', \auxd')$ defines the next state of the token after executing the transaction $T$. The next state hash is $\sthash'=H(\sthash,x)$.
\medskip

\noindent\textbf{Certifying a transaction} $T = (\sthash, D)$ involves the following steps:
\begin{enumerate}
\item $(\txhash,d) \gets \commit(H(D))$ is computed using a perfectly hiding commitment scheme $(\setup,\commit,\open)$.
The commitment $\txhash$ is called the \emph{transaction data hash}.
\item The hash value $h_T = H(\sthash, \txhash)$ is computed.
\item A solution $u$ is created such that $\predi(\systime_\mathsf{exp}, h_T, u) = 1$ for an expected\footnote{In order to avoid failed certification calls, users could (1) query the current time from $\unisrv$ to
  minimize the difference, and (2) prioritize safety in predicate design. For example, the delayed
  execution predicate $\mathsf{dex}_{\pubkey,\tau_0}$ from the introduction uses an inequality
  $\systime \geq \tau_0$, which remains satisfiable for all future times.} value $\systime_\mathsf{exp}$ of system time.
\item The query $Q = (\predi,\sthash,\txhash,u)$ is created.
\item $\unisrv$ is called to obtain $\pi\gets \unisrv(Q)$.
\item The certified transaction $(T,u,\txhash,d,\pi)$ is formed.
\end{enumerate}

\noindent\textbf{Verifying a certified transaction} A certified transaction $(T,u,\txhash,d,\pi)$ is verified in the state $(\predi,h)$ by the following algorithm:\medskip

\noindent $\certver(T,u,\txhash,d,\pi;\predi,h)$:
If at least one of the following checks fail, return 0, otherwise return 1:
\begin{enumerate}
\item $T.\sthash=h$;
\item $\open(\txhash,d)=H(T.D)$;
\item $\predi(\exttime(\pi), H(\sthash,\txhash),u)=1$;
\item $\univer(H(\predi,T.\sthash),\txhash,\pi)=1$.
\end{enumerate}

\begin{definition}[certification in a state]\label{de:certstate-pred}
A tuple $(T,u,\txhash,d,\pi)$ is said to be \emph{certified in state} $(\predi,h)$ iff $\certver(T,u,\txhash,d,\pi;\predi,h)=1$.
\end{definition}

Mint transactions are the same as in the signature-based ownership \cite{BTLR25}, i.e. they do not use generalized predicates.

\subsection{Token Ledger}

A \emph{token ledger} is a sequence
\[
(T_0, u_0, \pi_0; h^1_\mathsf{st}), (T_1, u_1,\txhash^1,d_1,\pi_1; h^2_\mathsf{st}), \ldots, (T_n, u_n, \txhash^n, d_n,\pi_n; h^{n+1}_\mathsf{st})
\]
where $(T_0, u_0, \pi_0)$ is a certified mint transaction and for every index $i=1,\ldots, n$:
\begin{enumerate}
\item $(T_{i}, u_{i}, \txhash^i,d_i, \pi_{i})$ is a certified transaction in the state $(T_{i-1}.D.\predi', h^{i}_\mathsf{st})$;
\item $h^{i}_\mathsf{st}=H(h^{i-1}_\mathsf{st},x_{i-1})$ where $x_{i-1} =T_{i-1}.D.x$.
\end{enumerate}

\section{Security}\label{sec:security}

Consider a token with the state $S = (\predi, \auxd)$. The transfer protocol ensures the following properties:
\begin{itemize}
\item \emph{No double-spending}:
Only one certified transaction can be created in the state $S$.
\item \emph{No association}: the Unicity service is unable to identify the transactions with the same token.
\item \emph{No blocking}:
Only those who can solve the predicate $\predi$ can block the state $S = (\predi, \auxd)$ if it was not blocked before.
\end{itemize}

\noindent In this section, we present security proofs for all three properties. The proofs of no double-spending and no association are very similar to the proofs in the paper \cite{BTLR25}. The no blocking property
had to be modified to cover arbitrary predicates.

\subsection{Security against Double-Spending}

A double-spending adversary uses $\unisrv$ as an oracle. \medskip

\noindent\textbf{Double-spending scenario} involves the following steps:
\begin{enumerate}
\item $(T,u,\txhash,d,\pinc),(T',u',\txhash',d',\pinc'),(\predi,h)\gets A^\unisrv$.
\item The attack is successful iff $T\neq T'$ and
\begin{equation}\label{eq:dscond-pred}
\certver(T,u,\txhash,d,\pinc;\predi,h)=\certver(T',u',\txhash',d',\pinc';\predi,h)=1 \enspace.
\end{equation}
\end{enumerate}

\begin{definition}[Double-spending security]
The Unicity Service is said to be $S$-secure against double-spending if it has $S$ as a security profile\footnote{The concept of security profiles is defined in \cite{BTLR25}.} in the double-spending scenario.
\end{definition}

\noindent\textbf{Analysis}: If the adversary is successful, then
from (\ref{eq:dscond-pred}) and the definition of $\certver$ it follows that
$T.\sthash = T'.\sthash=h$ and:
\begin{eqnarray*}
\univer(H(\predi, h), \txhash; \pinc) = \univer(H(\predi, h), \txhash'; \pinc') = 1\enspace,
\end{eqnarray*}
which implies $\txhash=\txhash'$ by equation~(\ref{eq:eqtx-pred}).
From Def.~\ref{de:certstate-pred} it also follows that
$\open(\txhash,d)=H(T.D)$ and $\open(\txhash,d')=\open(\txhash',d')=H(T'.D)$.
From $(h,T.D)=(T.\sthash,T.D)=T\neq T'=(T'.\sthash,T'.D)=(h,T'.D)$ it follows that $T.D\neq T'.D$. Hence, we have two cases:
\begin{itemize}
\item[a)] $H(T.D)=H(T'.D)$, which
means that a collision has been found for $H$.
\item[b)] $H(T.D)\neq H(T'.D)$, which implies $\open(\txhash,d)=H(T.D)\neq H(T'.D)=\open(\txhash,d')$ and hence, the commitment $\txhash$ has been opened in two different ways.
\end{itemize}

\begin{theorem}
If $H$ is $S$-secure collision-resistant and the commitment scheme is $S$-secure computationally binding, then the Unicity service is $S_\mathsf{double}$-secure against double-spending, where
$S_\mathsf{double}(\epsilon) = \frac{S(\epsilon/2)}{t_\mathsf{ver}+1}$
and $t_\mathsf{ver}$ is the predicate verification time.
\end{theorem}
\begin{proof}
Let $A$ be a $t$-time double-spending adversary that succeeds with probability $\epsilon$. We construct a collision-finder $A_\mathsf{coll}$ for the hash function and a double-opening adversary $A_\mathsf{com}$ for the commitment scheme as follows:
\begin{itemize}
\item $A_\mathsf{coll}$ proceeds as follows:
\begin{enumerate}
\item Simulate $(T,u,\txhash,d,\pinc),(T',u',\txhash',d',\pinc'),(\predi,h)\gets A^\unisrv$ by maintaining its own version of $\unisrv$.
\item Output the pair $(T.D,T'.D)$.
\end{enumerate}
The computational overhead function of $A_\mathsf{coll}$ is $\tau_\mathsf{coll}(t) = (t_\mathsf{ver}+1) \cdot t$, because simulating a $\unisrv$ query requires one predicate verification and the number of calls is limited by the running time $t$ of $A$.
\item $A_\mathsf{com}$ proceeds as follows:
\begin{enumerate}
\item Simulate $(T,u,\txhash,d,\pinc),(T',u',\txhash',d',\pinc'),(\predi,h)\gets A^\unisrv$ by maintaining its own version of $\unisrv$.
\item Output the triple $(\txhash,d,d')$.
\end{enumerate}
The computational overhead function of $A_\mathsf{com}$ is the same as
that of $A_\mathsf{coll}$, i.e. $\tau_\mathsf{com}(t)= \tau_\mathsf{coll}(t)=(t_\mathsf{ver}+1) \cdot t$.
\end{itemize}
 If $A$ succeeds, then in case a) the collision finder $A_\mathsf{coll}$ succeeds, and in case b) the double-opener $A_\mathsf{com}$ succeeds. Hence, $\epsilon \le \epsilon_\mathsf{coll} + \epsilon_\mathsf{com}$, where $\epsilon_\mathsf{coll}$ is the success probability of $A_\mathsf{coll}$ and
$\epsilon_\mathsf{com}$ is the success probability of $A_\mathsf{com}$.
Therefore, the function $S_\mathsf{double}$ defined by
\[
S_\mathsf{double}(\epsilon) = \tau^{-1}_\mathsf{coll}(S(\epsilon/2)) =
\frac{S(\epsilon/2)}{t_\mathsf{ver}+1}
\]
is a security profile of the Unicity service against double-spending.
\end{proof}

\subsection{Security against Association}

The Association adversary $A=(A_1,A_2)$ is two-stage. \medskip

\noindent\textbf{Association scenario} involves the following steps:
\begin{enumerate}
\item $(\sthash, \predi', \auxd', a)\gets A_1$.
\item $x\gets \{0,1\}^\ell$.
\item $\sthash'\gets H(\sthash,x)$.
\item $\txhash\gets\commitc(H(\predi',x,\auxd')))$.  \footnote{We denote by $\commitc(\param,m)$ the function that computes $(c,d)\gets \commit(\param; m)$ and returns only $c$.}
\item $x'\gets A_2(a; \sthash',\txhash)$.
\item The attack is successful iff $\sthash\in\{0,1\}^k$, $x'\in\{0,1\}^\ell$, and $H(\sthash,x')=\sthash'$. The success $\epsilon$ of $A$ is the probability that the attack is successful.
\end{enumerate}

\begin{definition}[association security]
The Unicity Service is said to be $S$-secure against association if it has $S$ as a security profile in the association scenario.
\end{definition}

\begin{theorem}
If the hash function is $S$-secure $(k,\ell)$-one-way\footnote{Defined in \cite{BTLR25}.} and the commitment scheme is perfectly hiding, then the Unicity Service is $S_\mathsf{assoc}$-secure against association, where $S_\mathsf{assoc}(\epsilon)= S(\epsilon) - t_\mathsf{sm} - t_\mathsf{hash} - t_\mathsf{com}$, where
$t_\mathsf{sm}$, $t_\mathsf{hash}$, $t_\mathsf{com}$ are the random sampling time, the hashing time, and the commitment computation time, respectively.
\end{theorem}
\begin{proof}
Let $A=(A_1,A_2)$ be a $t$-time adversary that succeeds in the association scenario with probability $\epsilon$. Consider the following modified attack scenario:
\begin{enumerate}
\item $(\sthash, \predi', \auxd', a)\gets A_1$.
\item $x\gets \{0,1\}^\ell$.
\item $\sthash'\gets H(\sthash,x)$.
\item $x''\gets \{0,1\}^\ell$.
\item $\txhash'\gets\commitc(H(\predi',x'',\auxd')))$.
\item $x'\gets A_2(a; \sthash',\txhash')$.
\item The attack is successful iff $\sthash\in\{0,1\}^k$, $x'\in\{0,1\}^\ell$, and $H(\sthash,x')=\sthash'$.
\end{enumerate}
%For any fixed value of $L=(\sthash, \pubkey', \auxd', a)$, due to the perfect hiding, the probability distributions of $\txhash$ and $\txhash'$ are equal and due to Lemma~\ref{le:outputindependence} (by taking $g(x)=H(\pubkey',x,\auxd'))$), the random variables $\txhash$ and $\txhash'$ are both independent of $x$. This means that the arguments $(\sthash',\txhash)$ and $(\sthash',\txhash')$ of $A(a;\cdot)$ are equal
%in both scenarios, and hence $A$ succeeds in the modified scenario with probability $\epsilon$. We construct an adversary $A'=(A'_1,A'_2)$ as follows:

\noindent For any fixed value of $L=(\sthash, \predi', \auxd', a)$, due to perfect hiding, commitments $\txhash=\commitc(H(\predi',x,\auxd'))$ and $\txhash'=\commitc(H(\predi',x'',\auxd'))$ have equal probability distributions. Moreover, by the Output Independence Lemma\footnote{Defined in \cite{BTLR25}.} (with $g(x)=H(\predi',x,\auxd')$), the random variables $x$ and $\txhash=\commitc(H(\predi',x,\auxd'))$ are independent.

Since $x''$ and $x$ are independent, the commitment $\txhash'=\commitc(H(\predi',x'',\auxd'))$ is independent of both $x$ and $\sthash'=H(\sthash,x)$. Therefore, the joint distributions of $(\sthash',\txhash)$ and $(\sthash',\txhash')$ are equal, and hence $A$ succeeds in the modified scenario with probability $\epsilon$. We construct an adversary $A'=(A'_1,A'_2)$ as follows:
\begin{itemize}
\item $A'_1$ proceeds as follows:
\begin{enumerate}
\item $(\sthash, \predi', \auxd', a)\gets A_1$;
\item return $(\sthash, a')$, where $a'=(\predi', \auxd', a)$.
\end{enumerate}
\item $A'_2(a';y)$ with $a'=(\predi', \auxd', a)$ proceeds as follows:
\begin{enumerate}
\item $x''\gets \{0,1\}^\ell$;
\item $\txhash'\gets\commitc(H(\predi',x'',\auxd')))$;
\item $x'\gets A_2(a; y,\txhash')$;
\item return $x'$.
\end{enumerate}
\end{itemize}

\noindent The computational time overhead function of $A'$ is $\tau(t) = t + t_\mathsf{sm} + t_\mathsf{hash} + t_\mathsf{com}$ and hence, the function $S_\mathsf{assoc}$ defined by
\[
S_\mathsf{assoc}(\epsilon) = \tau^{-1}(S(\epsilon)) = S(\epsilon) - t_\mathsf{sm} - t_\mathsf{hash} - t_\mathsf{com}
\]
is a security profile of the Unicity service against association.
\end{proof}

\subsection{Security against Blocking}

By a predicate family we mean a pair $(\prgen,\prsig)$ of algorithms so that:
\begin{itemize}
\item $(\prikey,\predi)\gets\prgen$ generates a private key $\prikey$ and a predicate $\predi$.
\item $\prsig(\prikey,m)$ solves the predicate for $m$, i.e. either $\bot\gets\prsig(\prikey,m)$ (the solver gives up) or $(\systime,u)\gets\prsig(\prikey,m)$ such that $\predi(\systime,m,u)=1$.
\end{itemize}
The case $\bot\gets\prsig(\prikey,m)$ is necessary because the predicates can potentially be chosen so that they cannot be satisfied.

\medskip
\noindent A predicate solving adversary $A_\mathsf{solve}^\mathcal{O}$ for the predicate family $(\prgen,\prsig)$ uses an oracle $\mathcal{O}$
that uses (initially empty) dictionaries $\prikey[\iota]$, $\predi[\iota]$ as the state and answers to two types of queries:
\begin{itemize}
\item $\mathcal{O}(\mathsf{gen};\iota)$ -- a generation query that returns $\bot$ if $\predi[\iota]\neq \bot$, and otherwise generates $(\prikey,\predi)\gets\prgen$, sets $\prikey[\iota]\gets \prikey$, $\predi[\iota]\gets\predi$, and
returns $\predi$.
\item $\mathcal{O}(\mathsf{solve};\iota,m)$ -- a solving query that returns
$\bot$ if either $\predi[\iota]=\bot$ or $\prsig(\prikey[\iota],m)=\bot$, and otherwise, if $(\systime,u)\gets \prsig(\prikey[\iota],m)$, it returns $(\systime,u)$.
\end{itemize}
The predicate solving scenario involves the following steps:
\begin{enumerate}
\item $(\iota,\systime,m,u)\gets A_\mathsf{solve}^\mathcal{O}$
\item The attack is successful if:
\begin{itemize}
\item[a)] $\predi[\iota]\neq\bot$, i.e. the query $\mathcal{O}(\mathsf{gen};\iota)$ was made by $A_\mathsf{solve}^\mathcal{O}$.
\item[b)] $\predi[\iota](\systime,m,u)=1$
\item[c)] All queries of the form $\mathcal{O}(\mathsf{solve};\iota,m)$ made by  $A_\mathsf{solve}^\mathcal{O}$ (if there were any) were answered with $\bot$.
\end{itemize}
\end{enumerate}

\noindent A blocking adversary $A$ uses two oracles:
\begin{enumerate}
\item $\mathsf{US}$: the Unicity Service,
\item $\mathsf{TS}$: that uses (initially empty) dictionaries $\prikey[\iota]$, $\predi[\iota]$ as the state and answers to two types of queries:
\begin{itemize}
\item $\mathsf{TS}(\mathsf{gen};\iota)$ -- a generation query that returns $\bot$ if $\prikey[\iota]\neq \bot$, and otherwise generates $(\prikey,\predi)\gets\prgen$, sets $\prikey[\iota]\gets \prikey$, $\predi[\iota]\gets\predi$, and
returns $\predi$.
\item $\mathsf{TS}(\mathsf{solve};\iota,h,D)$ -- a solving query that returns
$\bot$ if either $\prikey[\iota]=\bot$ or $\prsig(\prikey[\iota],H(h,\txhash))=\bot$, and otherwise, returns $(\systime,u, \txhash, d)$, where
$(\txhash, d)\gets \commit(H(D))$ and
$(\systime,u)\gets \prsig(\prikey[\iota],H(h,\txhash))$.
\end{itemize}
\end{enumerate}

\noindent\textbf{Blocking scenario} involves the following steps:
\begin{enumerate}
\item $(\iota,\sthash)\gets A^{\mathsf{US},\mathsf{TS}}$
\item $A$ is successful if:
\begin{itemize}
\item[a)] $\predi[\iota]\neq\bot$
\item[b)] $R[H(\predi[\iota],\sthash)]\neq\bot$ after the scenario
\item[c)] No (successful) queries of the form $\mathsf{TS}(\mathsf{solve};\iota,\sthash,D)$ were made.
\end{itemize}
The success $\epsilon$ of $A$ is the probability that the attack is successful
\end{enumerate}

\noindent Note that if such a query was made, then the request $Q=(\predi[\iota],\sthash, \txhash,u)$ to $\mathsf{US}$ at system time $\systime$ will trivially ensure $R[H(\predi[\iota],\sthash)]\neq\bot$.
Note that the adversary can set the system time appropriately before the request $Q$. Hence, this is excluded by the security condition. \medskip

\begin{definition}[blocking security]
The Unicity Service is said to be $S$-secure against blocking if it has $S$ as a security profile in the blocking scenario.
\end{definition}

\noindent\textbf{Analysis}: The adversary $A$ can be successful in the following cases:
\begin{itemize}
\item[a)] A request $Q=(\predi', h'_\mathsf{st}, h_\mathsf{tx}, u)$ with $(\predi',h'_\mathsf{st})\neq (\predi[\iota],\sthash)$ to $\mathsf{US}$ enforces $R[H(\predi[\iota], \sthash)]\neq\bot$, which means that  $H(\predi[\iota],\sthash)=H(\predi', h'_\mathsf{st})$ and hence, a collision for $H$ was found.
\item[b)] A request $Q=(\predi[\iota], \sthash, h_\mathsf{tx}, u)$ to $\mathsf{US}$ enforces $R[H(\predi[\iota], \sthash)]\neq\bot$, which implies $\predi[\iota](\systime, H(\sthash,h_\mathsf{tx}), u)=1$ from the description of $\mathsf{US}$. Then we have two possibilities:
\begin{itemize}
\item[b1)] A query $(\systime',u', \txhash', d)\gets\mathsf{TS}(\mathsf{solve};\iota,\sthash',D)$ was made such that the equality
$H(\sthash',\txhash')=H(\sthash,\txhash)$ holds. From the success condition of $A$ it follows that  $\sthash'\neq \sthash$ and we have a collision for $H$.
\item[b2)] If no queries $(\systime',u', \txhash', d)\gets\mathsf{TS}(\mathsf{solve};\iota,\sthash',D)$ were made with $H(\sthash',\txhash')=H(\sthash,\txhash)$ then this means that $A$ was able to
solve the predicate family, i.e. for $m=H(\sthash,\txhash)$ finds $\tau,u$
so that $\predi[\iota](\tau,m,u)=1$ without using the predicate solving functionality.
\end{itemize}
\end{itemize}

\begin{theorem}
If the predicate family $(\prgen,\prsig)$ is $S$-secure against solving and the hash function is $S$-secure collision-resistant, then the Unicity service is
$S_\mathsf{block}$-secure against blocking, where $S_\mathsf{block}(\epsilon) =
\frac{S(\epsilon/2)}{\max\{t_\mathsf{gen}, t_\mathsf{sig}, t_\mathsf{ver},t_\mathsf{com},t_\mathsf{hash}\}}$
and
$t_\mathsf{gen}$, $t_\mathsf{sig}$, $t_\mathsf{ver},t_\mathsf{com},t_\mathsf{hash}$ are the
key generation time (for $\prgen$), solving time (for $\prsig$), verification time (for $\predi$), commitment time, and hashing time, respectively.
\end{theorem}
\begin{proof}
Let $A$ be a $t$-time blocking adversary that succeeds with probability $\epsilon$. We construct a collision-finder $A_\mathsf{coll}$ and a
solver $A^\mathcal{O}_\mathsf{solve}$ as follows:
\begin{itemize}
\item $A_\mathsf{coll}$ proceeds as follows:
 \begin{enumerate}
   \item Simulates $(\iota,\sthash)\gets A^{\mathsf{US},\mathsf{TS}}$ and records all the oracle queries.
   \item If $A^{\mathsf{US},\mathsf{TS}}$ was successful and either the case a) or b1) occurs, $A_\mathsf{coll}$ outputs the collision that is guaranteed in this case.
 \end{enumerate}
The oracles are simulated as follows:
 \begin{itemize}
   \item $\mathsf{US}$-queries: $A_\mathsf{coll}$ maintains its own version of $R$.
   \item $\mathsf{TS}$-queries: $A_\mathsf{coll}$ directly uses $\prgen$ and $\prsig$.
\end{itemize}
The computational time overhead function for the construction of $A_\mathsf{coll}$ is $\tau_\mathsf{coll}(t) = \max\{t_\mathsf{gen}, t_\mathsf{sig}, t_\mathsf{ver}\}\cdot t$, where
$t_\mathsf{ver}$ is the predicate verification time (for $\unisrv$-queries),
$t_\mathsf{sig}$ is the predicate solving time (for $\mathsf{TS}(\mathsf{solve};\cdot)$-queries), and $t_\mathsf{gen}$ is the generation time (for $\mathsf{TS}(\mathsf{gen};\cdot)$-queries).

\item $A_\mathsf{solve}^{\mathcal{O}}$ proceeds as follows:
 \begin{enumerate}
  \item Simulates $(\iota,\sthash)\gets A^{\mathsf{US},\mathsf{TS}}$ and records all the oracle queries.
  \item If $A^{\mathsf{US},\mathsf{TS}}$ was successful and b2) occurs and $Q=(\predi', \sthash, h_\mathsf{tx}, u)$ was the request that enforces $R[H(\predi[\iota], \sthash)]\neq\bot$ at system time $\tau$ then:
  \begin{enumerate}
    \item[3.] $m\gets H(\sthash, h_\mathsf{tx})$.
    \item[4.] Output $(\iota,\tau,m,u)$.
  \end{enumerate}
 \end{enumerate}
The oracles are simulated as follows:
 \begin{itemize}
  \item $\mathsf{US}$-queries are simulated so that $A_\mathsf{solve}^\mathcal{O}$ maintains its own version of $R$.
  \item $\mathsf{US}$-queries are simulated so that $A_\mathsf{solve}^\mathcal{O}$ maintains its own copy of the dictionary $\predi[]$ and processes the queries as follows:
  \begin{itemize}
      \item $\mathsf{TS}(\mathsf{gen};\iota)$ -- If $\predi[\iota]\neq \bot$ then return $\bot$. Otherwise, call $\predi[\iota]
      \gets\mathcal{O}(\mathsf{gen};\iota)$ and return $\predi[\iota]$.
      \item $\mathsf{TS}(\mathsf{solve};\iota,h,D)$ -- If $\predi[\iota]=\bot$ then return $\bot$. Otherwise, compute $(\txhash,d)\gets\commit(H(D))$, call $v\gets \mathcal{O}(\mathsf{solve};\iota, H(h,\txhash))$ and return $v$. Note that either $v=\bot$ or $v=(\systime,u)$.
  \end{itemize}
 \end{itemize}
 As the request $Q$ was accepted by $\unisrv$ and it changes $R[H(\predi[\iota],\sthash)]$, we have $\predi'=\predi[\iota]$ (description of $\unisrv$) and
 $\predi[\iota](\systime,m,u)=1$. Note that in the case b2) the request $\mathsf{TS}(\mathsf{solve}; \iota, \sthash,D)$ with $\open(\txhash,d)=H(D)$ was never made which also means that the query $\mathcal{O}(\mathsf{solve};\iota,m)$ with $m=H(\sthash,\txhash)$ was never made, and hence,
 $A_\mathsf{solve}^{\mathcal{O}}$ is successful in the predicate solving scenario. The computational time overhead function for the construction of $A_\mathsf{solve}^\mathcal{O}$, on the assumption that $t\ge 1$, is
 \[
 \tau_\mathsf{solve}(t) = \max\{t_\mathsf{ver},t_\mathsf{com}\}\cdot t + t_\mathsf{hash}\le \max\{t_\mathsf{ver},t_\mathsf{com},t_\mathsf{hash}\}\cdot t\]
where
$t_\mathsf{ver}$ is the predicate verification time (for $\unisrv$ queries), $t_\mathsf{com}$ is the commitment time (for $\mathsf{TS}(\mathsf{solve};\cdot)$-queries) and
$t_\mathsf{hash}$ is the hash computation time (for output).
\end{itemize}
If $A$ succeeds, then either $A_\mathsf{coll}$ or $A_\mathsf{solve}^\mathcal{O}$ succeeds and hence $\epsilon \le \epsilon_\mathsf{coll} + \epsilon_\mathsf{solve}$.
As for the function $\tau(t)=\max\{t_\mathsf{gen}, t_\mathsf{sig}, t_\mathsf{ver},t_\mathsf{com},t_\mathsf{hash}\}\cdot t$ the inequalities $\tau_\mathsf{coll}(t)\le \tau(t)$
and $\tau_\mathsf{solve}(t)\le \tau(t)$ hold, we imply that
\[
S_\mathsf{block}(\epsilon) = \tau^{-1}(S(\epsilon/2)) =
\frac{S(\epsilon/2)}{\max\{t_\mathsf{gen}, t_\mathsf{sig}, t_\mathsf{ver},t_\mathsf{com},t_\mathsf{hash}\}}
\]
is a security profile of the Unicity Service against blocking.
\end{proof}

\section{Atomic Swap in the Unicity Infrastructure}\label{sec:swap}
In this section, we describe an atomic commit protocol in the Unicity Infrastructure that uses specific predicates. The protocol has several advantages compared to the hashed timelock based inter-blockchain atomic swaps used for example in Bitcoin-like blockchains. For example, it is secure even if unconditional delays occur in the system.
The key aspect of success is the presence of a common reference information provided by the Unicity service.

\subsection{Motivation and General Idea}
The goal of atomic swap is to fairly and securely exchange tokens between their owners. Assume users $A$ and $B$ own tokens $\mathsf{Tok}_A$ and $\mathsf{Tok}_B$ and they want to exchange their tokens so that after the exchange, $A$ owns $\mathsf{Tok}_B$ and vice versa. The parties agree to a certain timeout and run the swap protocol so that the following properties hold:
\begin{enumerate}
\item If both parties follow the protocol, then the tokens change their owners.
\item If at least one party deviates from the protocol, then $A$ and $B$ will retain control over their tokens after the agreed timeout.
\end{enumerate}
To achieve this goal, the following multi-stage protocol is used:
\begin{enumerate}
\item The parties transfer their tokens to special \emph{preparation states} and agree on the timeout.
\item The parties send each other the whole token ledgers and verify whether the tokens are indeed in the preparation states (based on the ledgers).
\item The parties (are supposed to) \emph{commit}, i.e. transfer (independently) their tokens (directly from the preparation states) to special \emph{commit (swap) states}.
\end{enumerate}
The swap states are defined (via suitable ownership predicates) so that:
\begin{enumerate}
\item If $\mathsf{Tok}_A$ is in the swap state, then $B$ will control (can make the next transaction with) $\mathsf{Tok}_A$ if $B$ committed $\mathsf{Tok}_B$ in time (before the timeout), and otherwise, $A$ will retain its control after the timeout.
\item If $\mathsf{Tok}_B$ is in the swap state, then $A$ will control (can make the next transaction with) $\mathsf{Tok}_B$ if $A$ committed $\mathsf{Tok}_A$ in time (before the timeout), and otherwise, $B$ will retain its control after the timeout.
\end{enumerate}
The main question is how to define the abstract ownership predicates for the preparation states and for the swap states that guarantee the requirements of the atomic swap protocol, and how to define the unlocking witness $u$ of the swap predicate so that it includes the knowledge about actions of the other party, considering that the commitments may happen independently without direct communication between the parties.

\subsection{Design Choices of the Protocol}

Assume that initially, $\mathsf{Tok}_A$ and $\mathsf{Tok}_B$ are in "pure" states
$S^A_0=(\nu_\mathsf{sig}(\pubkey_A;\cdot), h^0_A)$ and
$S^B_0=(\nu_\mathsf{sig}(\pubkey_B;\cdot),h^0_B)$, respectively, i.e. their use (next transfer) requires digital signatures of the parties. \medskip

\noindent In the \textbf{preparation phase}, the parties transfer their tokens to the preparation states $S^A_\mathsf{prep}=(\nu_\mathsf{prep}(\pubkey_A;\cdot),h^1_A)$ and $S^B_\mathsf{prep}=(\nu_\mathsf{prep}(\pubkey_B;\cdot),h^1_B)$, respectively by executing transactions $T^A_\mathsf{prep}=(h^0_A, D^A_\mathsf{prep})$ and
$T^B_\mathsf{prep}=(h^0_B, D^B_\mathsf{prep})$,
where:
\begin{eqnarray*}
D^A_\mathsf{prep}&=&(\nu_\mathsf{prep}(\pubkey_A;\cdot),x^A_\mathsf{prep},\mathsf{aux}^A_\mathsf{prep})\\
D^B_\mathsf{prep}&=&(\nu_\mathsf{prep}(\pubkey_B;\cdot),x^B_\mathsf{prep},\mathsf{aux}^B_\mathsf{prep})\enspace.
\end{eqnarray*}
The execution means that the parties sign those transactions, obtain the inclusion proofs $\pi^A_\mathsf{prep}$ and $\pi^B_\mathsf{prep}$ from $\unisrv$, and add certified transactions $C^A_\mathsf{prep}=(T^A_\mathsf{prep},\sigma^A_\mathsf{prep},\pi^A_\mathsf{prep})$ and $C^B_\mathsf{prep}=(T^B_\mathsf{prep},\sigma^B_\mathsf{prep},\pi^B_\mathsf{prep})$ to their ledgers. As the parties are supposed to exchange their ledgers, they can both see $T^A_\mathsf{prep}$ and $T^B_\mathsf{prep}$ and can check that the transactions contain the correct preparation predicates. Note also that the predicate $\nu_\mathsf{prep}$ is defined in a way that after preparation, both parties still have full control over their tokens, i.e. transferring the tokens further requires only the knowledge of their own private keys.
\medskip

\noindent In the \textbf{commit phase}, the parties transfer their tokens to the swap states:
\begin{eqnarray*}
S^A_\mathsf{swap}&=&(\nu_\mathsf{swap}(\pubkey_A,h^1_A,\pubkey_B,h^1_B,\tau_\mathsf{max};\cdot),h^2_A)\\
S^B_\mathsf{swap}&=&(\nu_\mathsf{swap}(\pubkey_B,h^1_B,\pubkey_A,h^1_A,\tau_\mathsf{max};\cdot),h^2_B)
\end{eqnarray*}
by executing transactions $T^A_\mathsf{swap}=(h^1_A, D^A_\mathsf{swap})$ and
$T^B_\mathsf{swap}=(h^1_B, D^B_\mathsf{swap})$, but the problem is that the parties may not see each other transactions. For example, if $A$ is malicious and commits, but after that does not send $T^A_\mathsf{swap}$ (or any other information)
to $B$. If $B$ also committed, then for making the next transaction with $\mathsf{Tok}_A$ (swap happened) $B$ has to construct a certified transaction for $T^A_\mathsf{swap}$ without knowing
the contents of the transaction. If $A$ did not commit and refuses to communicate with $B$ then how can $B$ "convince" the swap predicate of his own token that $A$ did not commit in time?

The only additional information source for $B$ is the Unicity Service, which gives information of type $R[k]=v$, i.e. which indices $k$ correspond to which values. For example, $B$ can query $\unisrv$ for the status of $k^A_\mathsf{prep}=H(S^A_\mathsf{prep})$, and
if $R[k^A_\mathsf{prep}]=\bot$ at $\tau'>\tau_\mathsf{max}$, then $B$ knows that $A$ did not commit in time because a timely swap transaction would have caused $R[k^A_\mathsf{prep}]\neq \bot$ before $\tau_\mathsf{max}$. However, if $R[k^A_\mathsf{prep}]=v\neq\bot$ at $\tau'>\tau_\mathsf{max}$ then we do not know whether the transaction $T=(h^1_A,D)$ that defined $R[k^A_\mathsf{prep}]$ (and spent $S^A_\mathsf{prep}$) was a correct swap transaction,
because on one hand, $v$ is computed using a perfectly hiding commitment scheme and contains no information about $D$.
On the other hand, even if swap transactions were certified without the commitment scheme (like minting transactions), then still $v=H(\nu,x,\mathsf{aux})$ and $x$ is random. So for $v$ giving any useful information about the correctness of $T$, we define the following requirements to swap transactions:
\begin{enumerate}
\item Swap transactions are certified without the commitment scheme (similar to mint transactions).
\item The fields $x$ and $\mathsf{aux}$ in swap transactions must be known constants. We choose these constants to be $x=0^\ell$ and $\mathsf{aux}=\bot$.
\end{enumerate}
This means that if $\pubkey_A,h^1_A,\pubkey_B,h^1_B,\tau_\mathsf{max}$ are known, then also the value
\[
v^A_\mathsf{swap}= H(\nu_\mathsf{swap}(\pubkey_A,h^1_A,\pubkey_B,h^1_B,\tau_\mathsf{max};\cdot),0^\ell,\bot)
\]
is known. Under these requirements,  $\mathsf{Tok}_A$ is correctly transferred from the preparation state
$S^A_\mathsf{prep}=(\nu_\mathsf{prep}(\pubkey_A;\cdot),h^1_A)$
to a correct swap state only if
$R[k^A_\mathsf{prep}]=v^A_\mathsf{swap}$.
Therefore, we have the following rollback rule:
\medskip\medskip

\noindent \textbf{Rollback rule}: If
%$R[k^A_\mathsf{prep}]=\bot$ before the swap protocol and
$R[k^A_\mathsf{prep}]\neq v^A_\mathsf{swap}$ at some point $\tau'>\tau_\mathsf{max}$, then $A$ certainly did not commit in time.
\medskip

\noindent We also need a positive rule to decide that $A$ committed in time. Assume now that $R[k^A_\mathsf{prep}]=v^A_\mathsf{swap}$ at some point $\tau'\le \tau_\mathsf{max}$.
Does this imply that $A$ has committed in time?
The answer turns out to be yes under the assumption that no $H$-collisions occur in the system and the signature scheme used is existentially unforgeable. This follows from the non-blocking security of the system -- only $A$ can block the state $S^A_\mathsf{prep}$ with a request $Q=(\nu_\mathsf{prep}(\pubkey_A;\cdot), h^1_A,v^A_\mathsf{swap},u)$ to $\unisrv$ and once $A$ has done it, the corresponding inclusion proof $\pi$ enables $A$ to construct a certified transaction that transfers $\mathsf{Tok}_A$ to the state $S^A_\mathsf{swap}$. It may seem that we are done, but not yet. In practice, we cannot assume that $B$ can always obtain the inclusion proof for $R[k^A_\mathsf{prep}]=v^A_\mathsf{swap}$ from $\unisrv$ in time (i.e., before $\tau_\mathsf{max}$) because $A$ may have committed in the last minute, $\unisrv$ may not always be accessible promptly, and the clocks of $B$ and $\unisrv$ may not be ideally synchronized. \medskip

\noindent So, what can we imply if we only have information that $R[k^A_\mathsf{prep}]=v^A_\mathsf{swap}$ at some point $\tau'> \tau_\mathsf{max}$? Similarly to the discussion in the last paragraph, we can conclude that if no $H$-collisions occurred in the system, then a request of the form $Q=(\nu_\mathsf{prep}(\pubkey_A;\cdot), h^1_A,v^A_\mathsf{swap},u)$ was made to $\unisrv$ that defines $R[k^A_\mathsf{prep}]=v^A_\mathsf{swap}$, but we are not sure if it occurred before the timeout $\tau_\mathsf{max}$. That is why we will define the predicate $\nu_\mathsf{prep}(\pubkey_A;\cdot)$ in such a way that transferring $\mathsf{Tok}_A$ from the state $(\nu_\mathsf{prep}(\pubkey_A;\cdot), h^1_A)$ to the corresponding swap state with the ownership predicate $\nu_\mathsf{swap}(\pubkey_A,h^1_A,\pubkey_B,h^1_B,\tau_\mathsf{max};\cdot)$ is only possible before timeout $\tau_\mathsf{max}$. This is also the main reason why we need special preparation states. Now we finally have a proper execution rule:\medskip

\noindent \textbf{Execution rule}: If
%$R[k^A_\mathsf{prep}]=\bot$ before the swap protocol and
$R[k^A_\mathsf{prep}]=v^A_\mathsf{swap}$ at any point, no collisions occurred and the signature scheme is existentially unforgeable, then $A$ committed in time.
\medskip

\noindent Figure~\ref{fi:aswap} depicts the details of the commitment step in the atomic swap protocol, assuming that both parties successfully commit.

\begin{figure}[ht]
\begin{center}
\includegraphics[width=12cm]{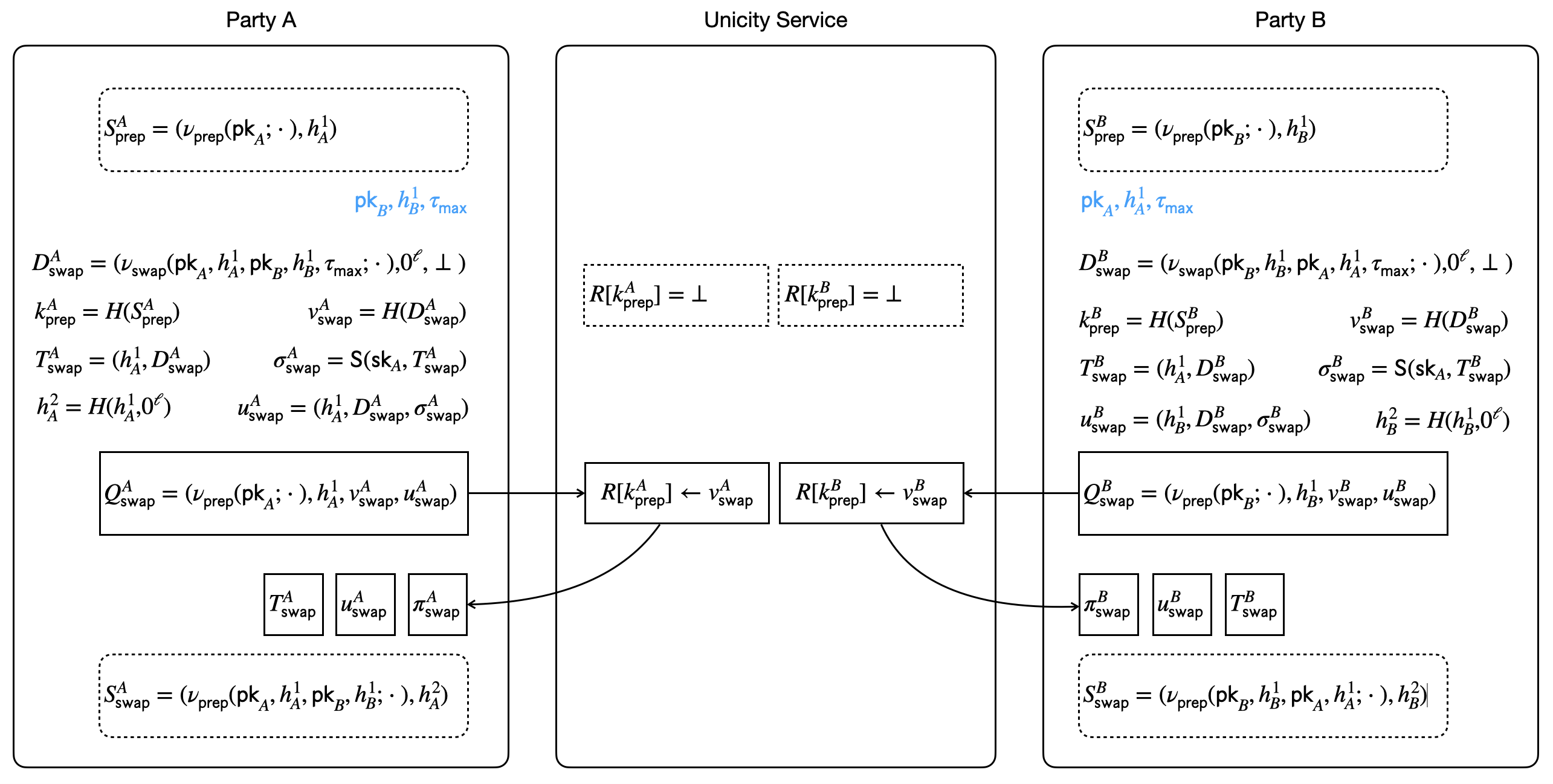}
\caption{Computations and message flow of the commit phase.}\label{fi:aswap}
\end{center}
\end{figure}

\subsection{Swap-Related Predicates}
In this section, we give definitions to the prepare predicate and the swap predicate and prove that their properties are sufficient to guarantee secure atomic swaps.

\subsubsection{Swap Preparation Predicate}

A special predicate $\nu_\mathsf{prep}(\pubkey;\cdot)$ that is used as an abstract ownership condition for preparing tokens for atomic swaps. Functionally, it is similar to the signature predicate, but has different unlock arguments $u=(h,D,\sigma)$, where $h$ is a hash value, $D$ is an arbitrary data structure, and $\sigma$ is a digital signature. The value of the predicate is defined as follows:\medskip

\noindent $\nu_\mathsf{prep}(\pubkey;\tau,m,h,D,\sigma) = 1$ iff:
\begin{enumerate}
\item $\sigver(\pubkey,m,\sigma)=1$
\item $m = H(h, H(D))$
\item If $D=(\nu_\mathsf{swap}(\pubkey,h,\pubkey',h',\tau_\mathsf{max};\cdot),0^\ell\!\!,\bot)$ (for some $\pubkey'\!,h'\!,\!\tau_\mathsf{max}$\!) then $\tau\!\le\!\tau_\mathsf{max}$.
%\begin{enumerate}
%\item[3.1.] $\pubkey_1=\pubkey$ ($T$ swaps a token controlled by $\pubkey$)
%\item[3.2.] $\tau\le \tau_\mathsf{max}$ ($T$ is executed in time)
%\item[3.3.] $x=0^\ell$
%\item[3.4.] $a=\bot$
%\end{enumerate}
\end{enumerate}
This definition implies that if $T=(h,D)$ is a swap transaction from the state $(\nu_\mathsf{prep}(\pubkey;\cdot),h)$ to a state with ownership predicate $\nu_\mathsf{swap}(\pubkey,h,\pubkey',h',\tau_\mathsf{max};\cdot)$, then it cannot be certified after the timeout $\tau_\mathsf{max}$ without $H$-collisions created in the system.

If $\pubkey_A,h^1_A,\pubkey_B,h^1_B,\tau_\mathsf{max}$ were the swap parameters agreed by $A$ and $B$ and $\mathsf{Tok}_A$ is in the state $S^A_\mathsf{prep}=(\nu_\mathsf{prep}(\pubkey_A;\cdot),h^1_A)$, then $A$ must commit at $\tau_\mathsf{max}$ ($\unisrv$ time) the latest, because later the request $Q^A_\mathsf{swap}=(\nu_\mathsf{prep}$$(\pubkey_A;\cdot),h^1_A,v^A_\mathsf{swap},u^A_\mathsf{swap})$ (Fig.~\ref{fi:aswap}) will not be accepted by $\unisrv$.

\subsubsection{Swap Predicate}

A special predicate $\nu_\mathsf{swap}(\pubkey,h,\pubkey',h',\tau_\mathsf{max};\cdot)$ that is used as an abstract ownership condition for tokens committed to swaps.
It has unlock arguments $u=(v,\pi,\sigma)$, where $v$ is a storage value (of $R$ in $\unisrv$), $\pi$ is an inclusion proof, and $\sigma$ is a digital signature.
The value of the predicate is defined as follows:\medskip

\noindent $\nu_\mathsf{swap}(\pubkey,h,\pubkey',h',\tau_\mathsf{max};\tau,m,v,\pi,\sigma) = 1$ iff:
\begin{itemize}
\item[1.] $f_\mathsf{inc}(H(S'_\mathsf{prep}),v,\pi)=1$ for $S'_\mathsf{prep}=(\nu_\mathsf{prep}(\pubkey';\cdot),h')$
\item[2.] For $v'_\mathsf{swap}=H(\nu_\mathsf{swap}(\pubkey',h',\pubkey,h,\tau_\mathsf{max};\cdot),0^\ell,\bot)$ either (a) or (b):
    \begin{itemize}
    \item[(a)] $\sigver(\pubkey',m,\sigma)\!=\!1 \wedge v=v'_\mathsf{swap}$
    \item[(b)] $\sigver(\pubkey,m,\sigma)\!=\!1 \wedge v\neq v'_\mathsf{swap}
    \wedge \tau_\mathsf{max}<\tau$
    \end{itemize}
\end{itemize}
\medskip

\noindent The inclusion proof verification function $f_\mathsf{inc}(k,v,\pi)\in\{0,1\}$ is evaluated via direct check $R[k]\stackrel{\mbox{?}}{=} v$ by $\unisrv$ and via the verification function $\mathcal{V}(k,v,\pi)$ by users.

In general, if $\mathsf{time}(\pi) \neq \tau$ then it is possible that users and $\unisrv$ (at $\tau$) evaluate $f_\mathsf{inc}(k,v,\pi)$ differently. However, it turns out that in the case of the swap protocol, this difference does not play any role.
%In a normal scenario, the swap predicate (with given arguments) is evaluated first at the $\unisrv$ side when the swap state is spent.

\begin{lemma}\label{le:atomic-rollback}
If $R[H(S'_\mathsf{prep})]\neq v'_\mathsf{swap}$ at $\tau>\tau_\mathsf{max}$, where $S'_\mathsf{prep}=(\nu_\mathsf{prep}(\pubkey';\cdot),h')$,  $v'_\mathsf{swap}=H(\nu_\mathsf{swap}(\pubkey',h',\pubkey,h,\tau_\mathsf{max};\cdot),0^\ell,\bot)$, and a token was in the state $S'_\mathsf{prep}$,
then the token was not properly committed
by its owner in the swap protocol with parameters
$\pubkey,h,\pubkey',h',\tau_\mathsf{max}$.\end{lemma}
\begin{proof}
Assume to the contrary that the owner of the token committed properly. This means that the transaction $T'_\mathsf{swap}=(h',D'_\mathsf{swap})$, where $D'_\mathsf{swap}=(\nu_\mathsf{swap}(\pubkey',h',\pubkey,h,\tau_\mathsf{max};\cdot),0^\ell,\bot)$ was certified at $S'_\mathsf{prep}$ before $\tau_\mathsf{max}$. This implies that there was a successful request $Q=(\nu_\mathsf{prep}(\pubkey';\cdot),h',v'_\mathsf{swap},u')$
to $\unisrv$ before $\tau_\mathsf{max}$,
but such a request will induce $R[H(S'_\mathsf{prep})]\gets v'_\mathsf{swap}\neq\bot$ before $\tau_\mathsf{max}$ and hence, as $R$ is append-only,
$R[H(S'_\mathsf{prep})]= v'_\mathsf{swap}$
at $\tau>\tau_\mathsf{max}$, a contradiction.
\end{proof}

\begin{lemma}\label{le:atomic-success}
If $R[H(S'_\mathsf{prep})]= v'_\mathsf{swap}$ at some point, where $S'_\mathsf{prep}=(\nu_\mathsf{prep}(\pubkey';\cdot),h')$,  $v'_\mathsf{swap}=H(\nu_\mathsf{swap}(\pubkey',h',\pubkey,h,\tau_\mathsf{max};\cdot),0^\ell,\bot)$, a token was in the state $S'_\mathsf{prep}$, and no $H$-collisions occur, then a successful request $Q=(\nu_\mathsf{prep}(\pubkey';\cdot),h',v'_\mathsf{swap},u')$
to $\unisrv$ was made at $\tau\le\tau_\mathsf{max}$ that
certifies the swap transaction $T'_\mathsf{swap}=(h',D'_\mathsf{swap})$, where $D'_\mathsf{swap}=(\nu_\mathsf{swap}(\pubkey',h',\pubkey,h,\tau_\mathsf{max};\cdot),0^\ell,\bot)$ at state $S'_\mathsf{prep}$.
\end{lemma}
\begin{proof}
If $R[H(S'_\mathsf{prep})]=v'_\mathsf{swap}$, then a successful request $\pi\gets \unisrv(\nu,h,v'_\mathsf{swap},u)$ at some point $\tau$ was received by $\unisrv$ so that $H(S'_\mathsf{prep})=H(\nu,h)$ and hence $\nu=\nu_\mathsf{prep}(\pubkey';\cdot)$ and $h=h'$. The request was successful then $u'$ is in the form $u'=(h,D,\sigma)$, and $\nu_\mathsf{prep}(\pubkey';\tau,H(h',v'_\mathsf{swap}),h,D,\sigma)=1$. From the definition of $\nu_\mathsf{prep}$ (2.), it follows that $m=H(h',v'_\mathsf{swap})=H(h',H(D))$ which implies $v'_\mathsf{swap}=H(D)$ and hence $D=(\nu_\mathsf{swap}(\pubkey',h',\pubkey,h,\tau_\mathsf{max};\cdot),0^\ell,\bot)$, which from the definition of $\nu_\mathsf{prep}$ (3.) implies
$\tau=\mathsf{time}(\pi)\le\tau_\mathsf{max}$.
As
\[
\nu_\mathsf{prep}(\pubkey';\mathsf{time}(\pi), H(h',v'_\mathsf{swap}), u') = 1\enspace.
\]
and $\univer(H(S'_\mathsf{prep}),v'_\mathsf{swap},\pi)=1$, we conclude that $(T'_\mathsf{swap},u',\pi)$ is a certified swap transaction at $S'_\mathsf{prep}$.
\end{proof}\medskip

\noindent If the signature scheme is existentially unforgeable and the private key of $\pubkey'$ is not compromised, then practically the conclusion of Lemma~\ref{le:atomic-success} means that the owner of the token committed properly.

The conditions $\mathcal{A}\equiv R[H(S'_\mathsf{prep})]=v'_\mathsf{swap}$
and $\mathcal{B}\equiv R[H(S'_\mathsf{prep})]\neq v'_\mathsf{swap}\wedge \tau_\mathsf{max}<\tau$ clearly cannot be both true at the same time. Moreover, if $\mathcal{B}$ holds at $\tau$, then $\mathcal{A}$ can never hold, because the proof of Lemma~\ref{le:atomic-success} suggests that if $\mathcal{A}$ is true at any time, then it must be true
also at some $\tau'\le \tau_\mathsf{max}$, but then, due to the append-only property, $\mathcal{A}$ must hold also at $\tau$.

Also, at any time $\tau>\tau_\mathsf{max}$, either $\mathcal{A}$ or $\mathcal{B}$ must be true, because always, either $R[H(S'_\mathsf{prep})]=v'_\mathsf{swap}$
or $R[H(S'_\mathsf{prep})]\neq v'_\mathsf{swap}$. Therefore, the swap predicate $\nu_\mathsf{swap}$ correctly implements the requirements of the swap protocol between $A$ and $B$:
\begin{itemize}
\item[1.] If $A$ and $B$ follow the protocol, then the swap happens
\item[2.] If $A$ or $B$ (or both) deviate from the protocol, then $A$ and $B$ will get their tokens back
\end{itemize}

\subsection{Detailed Message Flow of the Protocol}

In this subsection, we describe the message flow and computations during the swap protocol from the viewpoint of $A$.
\medskip

\noindent\textbf{Preparation of $A$}: To transfer $\mathsf{Tok}_A$ to the state $(\nu_\mathsf{prep}(\pubkey_A;\cdot),h^1_A)$, the party $A$ runs the following protocol that involves communication with $\unisrv$:
\begin{enumerate}
    \item $A$ computes: $D^A_\mathsf{prep}\gets (\nu_\mathsf{prep}(\pubkey_A;\cdot), 0^\ell,\bot)$
    \item $A$ computes: $v^A_\mathsf{prep}\gets H(D^A_\mathsf{prep})$
    \item $A$ computes: $T^A_\mathsf{prep}\gets (h^0_A, D^A_\mathsf{prep})$
    \item $A$ computes: $u^A_\mathsf{prep}\gets \sig(\prikey_A,H(h^0_A,v^A_\mathsf{prep}))$
    \item $A$ sends $\unisrv$: $Q^A_\mathsf{prep}=(\nu_\mathsf{sig}(\pubkey_A;\cdot),h^0_A,v^A_\mathsf{prep},u^A_\mathsf{prep})$
    \item $\unisrv$ computes: $\tau^A_\mathsf{prep}\gets \tau$
    \item $\unisrv$ computes: $k^A_\mathsf{sig}\gets H(\nu_\mathsf{sig}(\pubkey_A;\cdot),h^0_A)$
    \item $\unisrv$ checks: $R[k^A_\mathsf{sig}]=\bot$
    \item $\unisrv$ checks: $\nu_\mathsf{sig}(\pubkey_A;\tau^A_\mathsf{prep},H(h^0_A,v^A_\mathsf{prep}),u^A_\mathsf{prep})=1$
    \item $\unisrv$ obtains: $\pi^A_\mathsf{prep}$ for $Q^A_\mathsf{prep}$. Note that $\mathsf{time}(\pi^A_\mathsf{prep})=\tau^A_\mathsf{prep}$
    \item $\unisrv$ sets: $R[k^A_\mathsf{sig}]\gets (v^A_\mathsf{prep},u^A_\mathsf{prep})$
    \item $\unisrv$ sends $A$: $\pi^A_\mathsf{prep}$
    \item $A$ computes : $C^A_\mathsf{prep}\gets (T^A_\mathsf{prep},u^A_\mathsf{prep},v^A_\mathsf{prep},\pi^A_\mathsf{prep})$ and adds $C^A_\mathsf{prep}$ as a certified transaction to the ledger of $\mathsf{Tok}_A$.
\end{enumerate}
After these steps, $\mathsf{Tok}_A$ is in the state $(\nu_\mathsf{prep}(\pubkey_A;\cdot),h^1_A)$, where $h^1_A=H(h^0_A,0^\ell)$, and if
$B$ proceeds analogously, $\mathsf{Tok}_B$ is in the state $(\nu_\mathsf{prep}(\pubkey_B;\cdot),h^1_B)$, where $h^1_B=H(h^0_B,0^\ell)$. The parties now exchange the full token ledgers, check if the tokens are indeed in correct states, and agree on the timeout value $\tau_\mathsf{max}$. Party $A$ is now ready for the commitment phase. \medskip

\noindent\textbf{Commitment of $A$}: To transfer $\mathsf{Tok}_A$ to $(\nu^A_\mathsf{swap}(\pubkey_A,h^1_A,\pubkey_B,h^1_B,\tau_\mathsf{max};\cdot),h^2_A)$, the party $A$ runs the following protocol that involves communication with $\unisrv$:
\begin{enumerate}
    \item $A$ computes: $D^A_\mathsf{swap}\gets (\nu_\mathsf{swap}(\pubkey_A,h^1_A,\pubkey_B,h^1_B,\tau_\mathsf{max};\cdot), 0^\ell,\bot)$
    \item $A$ computes: $v^A_\mathsf{swap}\gets H(D^A_\mathsf{swap})$
    \item $A$ computes: $T^A_\mathsf{swap}\gets (h^1_A, D^A_\mathsf{swap})$
    \item $A$ computes: $\sigma^A_\mathsf{swap}\gets \sig(\prikey_A,H(h^1_A,v^A_\mathsf{swap}))$
    \item $A$ computes: $u^A_\mathsf{swap}\gets (h^1_A,D^A_\mathsf{swap},\sigma^A_\mathsf{swap})$
    \item $A$ sends $\unisrv$: $Q^A_\mathsf{swap}=(\nu_\mathsf{prep}(\pubkey_A;\cdot),h^1_A,v^A_\mathsf{swap},u^A_\mathsf{swap})$
    \item $\unisrv$ computes: $\tau^A_\mathsf{swap}\gets \tau$
    \item $\unisrv$ computes: $k^A_\mathsf{prep}\gets H(\nu_\mathsf{prep}(\pubkey_A;\cdot),h^1_A)$
    \item $\unisrv$ checks: $R[k^A_\mathsf{prep}]=\bot$
    \item $\unisrv$ checks: $\nu_\mathsf{prep}(\pubkey_A;\tau^A_\mathsf{swap},H(h^1_A,v^A_\mathsf{swap}),u^A_\mathsf{swap})=1$
    \item $\unisrv$ sets: $R[k^A_\mathsf{prep}]\gets (v^A_\mathsf{swap},u^A_\mathsf{swap})$
    \item $\unisrv$ obtains: $\pi^A_\mathsf{swap}$ for $Q^A_\mathsf{swap}$. Note that $\mathsf{time}(\pi^A_\mathsf{swap})=\tau^A_\mathsf{swap}$
    \item $\unisrv$ sends $A$: $\pi^A_\mathsf{swap}$
    \item $A$ computes : $C^A_\mathsf{swap}\gets (T^A_\mathsf{swap},u^A_\mathsf{swap},v^A_\mathsf{swap},\pi^A_\mathsf{swap})$ and adds $C^A_\mathsf{swap}$ as a certified transaction to the ledger of $\mathsf{Tok}_A$.
    \item $A$ sends $B$ : $C^A_\mathsf{swap}$
\end{enumerate}
$\mathsf{Tok}_A$ is now in state $(\nu_\mathsf{swap}(\pubkey_A,h^1_A,\pubkey_B,h^1_B,\tau_\mathsf{max};\cdot),h^2_A)$, where $h^2_A=H(h^1_A,0^\ell)$. If $B$ did the same then
$\mathsf{Tok}_B$ is in state $(\nu_\mathsf{swap}(\pubkey_B,h^1_B,\pubkey_A,h^1_A,\tau_\mathsf{max};\cdot),h^2_B)$, where $h^2_B=H(h^1_B,0^\ell)$
and $A$ receives $C^B_\mathsf{swap}=(T^B_\mathsf{swap},u^B_\mathsf{swap},v^B_\mathsf{swap},\pi^B_\mathsf{swap})$ from $B$. $A$ is now ready for successful finalization.\medskip

\noindent\textbf{Recovery of $A$}:  If $A$ has not received $C^B_\mathsf{swap}$ at $\tau_\mathsf{max}$, then:
\begin{enumerate}
    \item $A$ computes: $D^B_\mathsf{swap}\gets (\nu_\mathsf{swap}(\pubkey_B,h^1_B,\pubkey_A,h^1_A,\tau_\mathsf{max};\cdot), 0^\ell,\bot)$
    \item $A$ computes: $v^B_\mathsf{swap}\gets H(D^B_\mathsf{swap})$
    \item $A$ computes: $T^B_\mathsf{swap}\gets (h^1_B, D^B_\mathsf{swap})$
    \item $A$ computes: $k^B_\mathsf{prep}\gets H(\nu_\mathsf{prep}(\pubkey_B;\cdot),h^1_B)$
    \item $A$ sends $\unisrv$ : $Q_\mathsf{prf}(k^B_\mathsf{prep})$
    \item $\unisrv$ computes : $\tau'\gets \tau$
    \item $\unisrv$ computes : $v\gets R[k^B_\mathsf{prep}]$
    \item $\unisrv$ obtains : $\pi^B_\mathsf{prf}$ for $v=R[k^B_\mathsf{prep}]$.
    \item $\unisrv$ sends $A$ : $v,\pi^B_\mathsf{prf}$. (
    $\mathsf{time}(\pi^B_\mathsf{prf})=\tau'>\tau_\mathsf{max}$ and $\univer(k^B_\mathsf{prep},v,\pi^B_\mathsf{prf})=1$)
\end{enumerate}
If $v=(v^B_\mathsf{swap},u^B_\mathsf{swap})$ then $A$ computes
$\overline{C}^B_\mathsf{swap}=(T^B_\mathsf{swap},u^B_\mathsf{swap},v^B_\mathsf{swap},\pi^B_\mathsf{prf})$.
As $\univer(k^B_\mathsf{prep},v,\pi^B_\mathsf{prf})=1$ implies that
$R[k^B_\mathsf{prep}]=(v^B_\mathsf{swap},u^B_\mathsf{swap})$ at $\tau'=\mathsf{time}(\pi^B_\mathsf{prf})$, there was a request $Q=(\nu,\sthash,v^B_\mathsf{swap},u^B_\mathsf{swap})$ such that $H(\nu,\sthash)=k^B_\mathsf{prep}=H(\nu_\mathsf{prep}(\pubkey_B;\cdot),h^1_B)$
which implies $\nu=\nu_\mathsf{prep}(\pubkey_B;\cdot)$ and $\sthash=h^1_B$ if no $H$-collisions occur. Therefore, $\overline{C}^B_\mathsf{swap}$ is a certified transaction in $(\nu_\mathsf{prep}(\pubkey_B;\cdot),h^1_B)$ and can be used to extend the ledger of $\mathsf{Tok}_B$ (see also Lemma~\ref{le:atomic-success}),
and $A$ is now ready for successful finalization.\medskip

\noindent If $v=\bot$ or  $v=(v',u)$  where $v'\neq v^B_\mathsf{swap}$ then
$A$ is ready for unsuccessful finalization (Lemma~\ref{le:atomic-rollback}).
\medskip

\noindent\textbf{Successful finalization of $A$}: Party $A$ now controls $\mathsf{Tok}_B$.
The last certified transaction in the ledger of $\mathsf{Tok}_A$ is
$C^A_\mathsf{swap}=(T^A_\mathsf{swap},u^A_\mathsf{swap},v^A_\mathsf{swap},\pi^A_\mathsf{swap})$. The inclusion proof $\pi^A_\mathsf{swap}$ is the proof that $A$ have committed $\mathsf{Tok}_A$.
To transfer $\mathsf{Tok}_B$ from
$(\nu_\mathsf{swap}(\pubkey_B,h^1_B,\pubkey_A,h^1_A,\tau_\mathsf{max};\cdot),h^2_B)$
 to $(\nu^A_\mathsf{sig}(\pubkey_A;\cdot),h^3_A)$, the party $A$ runs the following protocol that involves communication with $\unisrv$:
\begin{enumerate}
    \item $A$ computes: $D^A_\mathsf{fin}\gets (\nu_\mathsf{sig}(\pubkey_A;\cdot), 0^\ell,\bot)$
    \item $A$ computes: $v^A_\mathsf{fin}\gets H(D^A_\mathsf{fin})$
    \item $A$ computes: $T^A_\mathsf{fin}\gets (h^2_B, D^A_\mathsf{fin})$
    \item $A$ computes: $\sigma^A_\mathsf{fin}\gets \sig(\prikey_A,H(h^2_B,v^A_\mathsf{fin}))$
    \item $A$ computes: $u^A_\mathsf{fin}\gets (v^A_\mathsf{swap},\pi^A_\mathsf{swap},\sigma^A_\mathsf{fin})$
    \item $A$ sends $\unisrv$: $Q^A_\mathsf{fin}=(\nu_\mathsf{swap}(\pubkey_B,h^1_B,\pubkey_A,h^1_A,\tau_\mathsf{max};\cdot),h^2_B,v^A_\mathsf{fin},u^A_\mathsf{fin})$
    \item $\unisrv$ computes: $\tau^A_\mathsf{fin}\gets \tau$
    \item $\unisrv$ computes: $k^A_\mathsf{fin}\gets H(\nu_\mathsf{swap}(\pubkey_B,h^1_B,\pubkey_A,h^1_A,\tau_\mathsf{max};\cdot),h^2_B)$
    \item $\unisrv$ checks: $R[k^A_\mathsf{fin}]=\bot$
    \item $\unisrv$ checks: $\nu_\mathsf{swap}(\pubkey_B,h^1_B,\pubkey_A,h^1_A,\tau_\mathsf{max};\tau^A_\mathsf{fin},H(h^2_B,v^A_\mathsf{fin}),u^A_\mathsf{fin})=1$.
    \item $\unisrv$ sets: $R[k^A_\mathsf{fin}]\gets (v^A_\mathsf{fin},u^A_\mathsf{fin})$
    \item $\unisrv$ obtains: $\pi^A_\mathsf{fin}$ for $Q^A_\mathsf{fin}$. Note that $\mathsf{time}(\pi^A_\mathsf{fin})=\tau^A_\mathsf{fin}$
    \item $\unisrv$ sends $A$: $\pi^A_\mathsf{fin}$
    \item $A$ computes : $C^A_\mathsf{fin}\gets (T^A_\mathsf{fin},u^A_\mathsf{fin},v^A_\mathsf{fin},\pi^A_\mathsf{fin})$ and adds $C^A_\mathsf{fin}$ as a certified transaction to the ledger of $\mathsf{Tok}_B$.
\end{enumerate}
$\mathsf{Tok}_B$ is now in the state $(\nu_\mathsf{sig}(\pubkey_A;\cdot),h^3_B)$, where $h^3_B=H(h^2_B,0^\ell)$, which means that $A$ has now full control over $\mathsf{Tok}_B$.
\medskip

\noindent\textbf{Unsuccessful finalization of $A$}: Party $A$ now controls $\mathsf{Tok}_A$.
The last record in the ledger of $\mathsf{Tok}_B$ is a certified transaction $C^B_\mathsf{swap}=(T^B_\mathsf{swap},u^B_\mathsf{swap},v^B_\mathsf{swap},\pi)$ (where $\pi\in\{\pi^B_\mathsf{prf},\pi^B_\mathsf{swap}\}$). In order to transfer the token $\mathsf{Tok}_A$ from the state
$(\nu_\mathsf{swap}(\pubkey_A,h^1_A,\pubkey_B,h^1_B,\tau_\mathsf{max};\cdot),h^2_A)$
to the state $(\nu^A_\mathsf{sig}(\pubkey_A;\cdot),h^3_A)$, the party $A$ runs the following protocol that involves communication with $\unisrv$:
\begin{enumerate}
    \item $A$ computes: $D^A_\mathsf{fin}\gets (\nu_\mathsf{sig}(\pubkey_A;\cdot), 0^\ell,\bot)$
    \item $A$ computes: $v^A_\mathsf{fin}\gets H(D^A_\mathsf{fin})$
    \item $A$ computes: $T^A_\mathsf{fin}\gets (h^2_A, D^A_\mathsf{fin})$
    \item $A$ computes: $\sigma^A_\mathsf{fin}\gets \sig(\prikey_A,H(h^2_A,v^A_\mathsf{fin}))$
    \item $A$ computes: $u^A_\mathsf{fin}\gets (v,\pi,\sigma^A_\mathsf{fin})$
    \item $A$ sends $\unisrv$: $Q^A_\mathsf{fin}=(\nu_\mathsf{swap}(\pubkey_A,h^1_A,\pubkey_B,h^1_B,\tau_\mathsf{max};\cdot),h^2_A,v^A_\mathsf{fin},u^A_\mathsf{fin})$
    \item $\unisrv$ computes: $\tau^A_\mathsf{fin}\gets \tau$
    \item $\unisrv$ computes: $k^A_\mathsf{fin}\gets H(\nu_\mathsf{swap}(\pubkey_A,h^1_A,\pubkey_B,h^1_B,\tau_\mathsf{max};\cdot),h^2_A)$
    \item $\unisrv$ checks: $R[k^A_\mathsf{fin}]=\bot$
    \item $\unisrv$ checks: $\nu_\mathsf{swap}(\pubkey_A,h^1_A,\pubkey_B,h^1_B,\tau_\mathsf{max};\tau^A_\mathsf{fin},H(h^2_A,v^A_\mathsf{fin}),u^A_\mathsf{fin})=1$
    \item $\unisrv$ sets: $R[k^A_\mathsf{fin}]\gets (v^A_\mathsf{fin},u^A_\mathsf{fin})$
    \item $\unisrv$ obtains: $\pi^A_\mathsf{fin}$ for $Q^A_\mathsf{fin}$. Note that $\mathsf{time}(\pi^A_\mathsf{fin})=\tau^A_\mathsf{fin}$
    \item $\unisrv$ sends $A$: $\pi^A_\mathsf{fin}$
    \item $A$ computes : $C^A_\mathsf{fin}\gets (T^A_\mathsf{fin},u^A_\mathsf{fin},v^A_\mathsf{fin},\pi^A_\mathsf{fin})$ and adds $C^A_\mathsf{fin}$ as a certified transaction to the ledger of $\mathsf{Tok}_B$.
\end{enumerate}
$\mathsf{Tok}_A$ is now in the state $(\nu_\mathsf{sig}(\pubkey_A;\cdot),h^3_A)$, where $h^3_A=H(h^2_A,0^\ell)$, which means that $A$ has now full control over $\mathsf{Tok}_A$.

\subsection{Atomic Swap with Hashed Timelock and Discussion}
Assume that $A$ owns $\mathsf{Tok}_A$
and $B$ owns $\mathsf{Tok}_B$.
Atomic swap protocol can also be designed with the hashed timelock predicate \cite{Herl18,Bitc18a,Bitc18b}  $\nu_\mathsf{htlc}(\pubkey,\pubkey',y,\tau_\mathsf{max};\cdot)$ as follows:
\begin{enumerate}
\item $A$ generates a random number $x$, computes $y\gets H(x)$ and sends $y$ to $B$.
\item $A$ and $B$ agree to two timeouts $\tau_A$ and $\tau_B$, where $\tau_A < \tau_B$ and the difference $\tau_B-\tau_A$ is sufficiently large.

\item $A$ transfers $\mathsf{Tok}_A$ to the lock state $S^A_\mathsf{htlc}=(\nu_\mathsf{htlc}(\pubkey_A,\pubkey_B,y,\tau_B;\cdot), h^1_A)$. In this state, $A$ will be able transfer the token back to herself after $\tau_B$
in case $B$ did not transfer the token before $\tau_B$. However, $B$ is not yet able to transfer the token because he does not know the pre-image $x$ of $y$.

\item $B$ checks that $\mathsf{Tok}_A$ is indeed in the state $S^A_\mathsf{htlc}$ and
transfers $\mathsf{Tok}_B$ to the state $S^B_\mathsf{htlc}=(\nu_\mathsf{htlc}(\pubkey_B,\pubkey_A,y,\tau_A;\cdot), h^1_B)$. In this state, $A$ is able to transfer $\mathsf{Tok}_B$ to her public key $\pubkey_A$ until the timeout $\tau_A$ and for transferring the token, she has to reveal $x$.

\item $A$ transfers $\mathsf{Tok}_B$ to the state with ownership $\nu_\mathsf{sig}(\pubkey_A;\cdot)$ and reveals $x$, which is included into the unlocking arguments $u=(x,\sigma_A)$.

\item $B$ checks the status of $\mathsf{Tok}_B$ via $\unisrv$ and obtains $x$ as a part of the unlocking arguments. He has now all information but a limited time until $\tau_B$ to transfer $\mathsf{Tok}_A$ to himself.

\item $B$ transfers $\mathsf{Tok}_A$ to the state with ownership $\nu_\mathsf{sig}(\pubkey_B;\cdot)$.
\end{enumerate}

\noindent This protocol has more steps compared to the symmetric atomic swap described above and, moreover, it involves a risk that if $B$ is late to transfer $\mathsf{Tok}_A$ at step 7 of the protocol (which may happen due to circumstances that $B$ does not control, such as network connection failures), then after $\tau_B$, party $A$ controls both tokens.

The conclusion is that the symmetric atomic swap protocol is better suited for the Unicity framework compared to the hashed timelock protocol, although we do not claim that the symmetric protocol would provide a better alternative to hashed timelock in an inter-blockchain setting where there are no reliable common reference services.

\pagebreak
\appendix

\section{Sequence Diagram with Predicates}
\begin{figure}[!h]
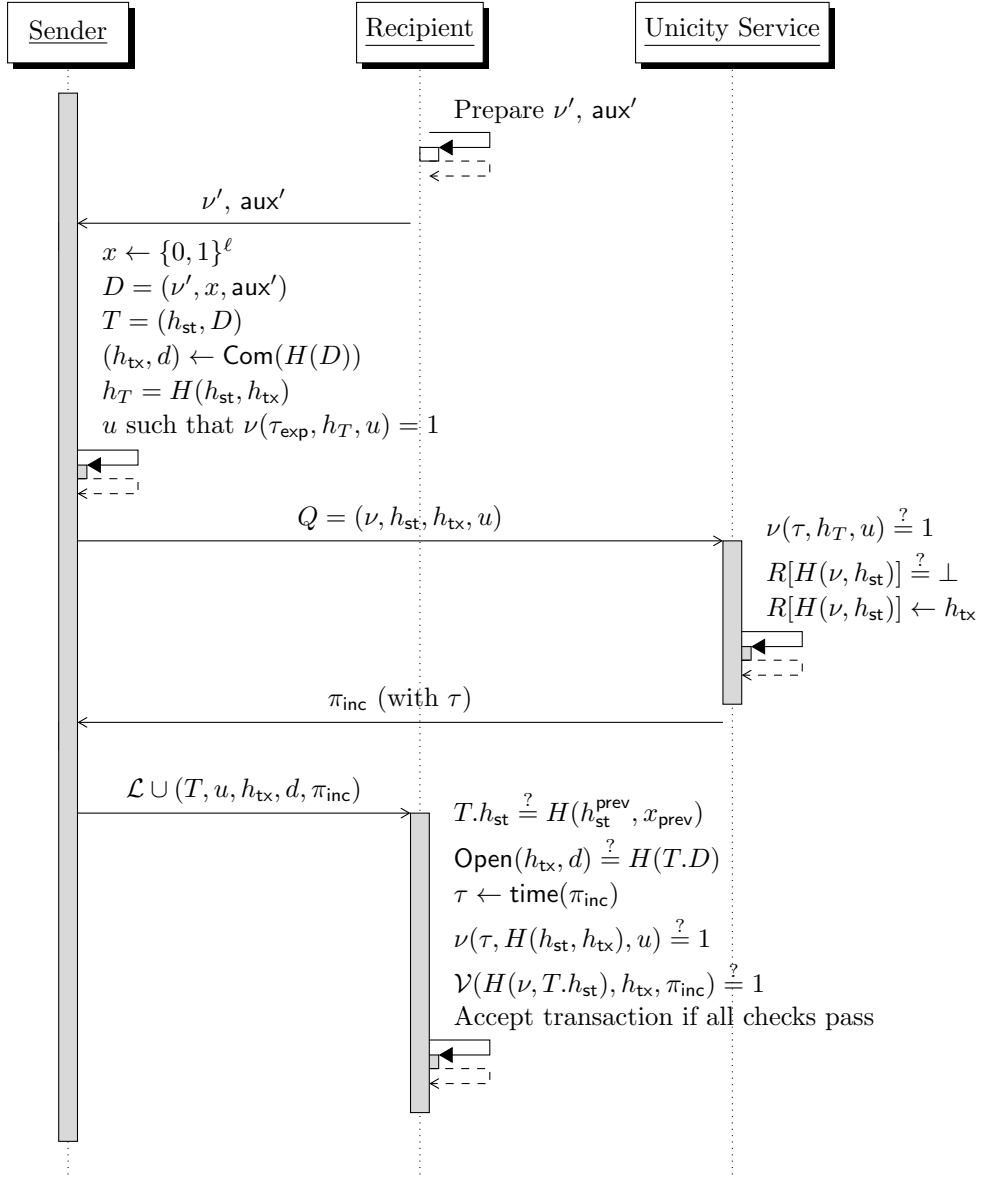

    \begin{center}
        \begin{sequencediagram}
            \newthread{S}{Sender}
            \newinst[3]{R}{Recipient}
            \newinst[2]{US}{Unicity Service}

            % Recipient prepares predicate and auxiliary data
            \begin{call}{R}{Prepare $\predi'$, $\auxd'$}{R}{}
            \end{call}

            % Recipient sends predicate to sender
            \begin{messcall}{R}{$\predi'$, $\auxd'$}{S}
            \end{messcall}

            % Sender creates transaction data
            \postlevel
            \postlevel
            \postlevel
            \begin{call}{S}{\shortstack[l]{
                    $x \gets \{0,1\}^\ell$     \\
                    $D = (\predi', x, \auxd')$ \\
                    $T = (\sthash, D)$         \\
                    $(\txhash, d) \gets \commit(H(D))$ \\
                    $h_T = H(\sthash, \txhash)$ \\
                    $u$ such that $\predi(\systime_\mathsf{exp}, h_T, u) = 1$
                }
                }{S}{}
            \end{call}

            % Sender requests certification from Unicity Service
            \begin{messcall}{S}{$Q = (\predi, \sthash, \txhash, u)$}{US}
                \postlevel
                \begin{call}{US}{\shortstack[l]{
                    $\predi(\systime, h_T, u) \stackrel{?}{=} 1$ \\
                    $R[H(\predi, \sthash)] \stackrel{?}{=} \bot$ \\
                    $R[H(\predi, \sthash)] \gets \txhash$
                    }}{US}{}
                \end{call}
            \end{messcall}
            \prelevel
            \begin{messcall}{US}{$\pinc$ (with $\systime$)}{S}
            \end{messcall}

            % Sender sends certified transaction to recipient
            \begin{messcall}{S}{$\mathcal{L} \cup (T, u, \txhash, d, \pinc)$}{R}
                \postlevel
                \postlevel
                \postlevel
                \postlevel
                \begin{call}{R}{\shortstack[l]{
                    $T.\sthash \stackrel{?}{=} H(h^\textsf{prev}_\mathsf{st},x_\textsf{prev})$ \\
                    $\open(\txhash, d) \stackrel{?}{=} H(T.D)$ \\
                    $\systime \gets \exttime(\pinc)$ \\
                    $\predi(\systime, H(\sthash, \txhash), u) \stackrel{?}{=} 1$ \\
                    $\univer(H(\predi, T.\sthash), \txhash, \pinc) \stackrel{?}{=} 1$ \\
                    Accept transaction if all checks pass
                    }}{R}{}
                \end{call}
            \end{messcall}
        \end{sequencediagram}
        \caption{One unicity transaction, where the predicate $\predi$ represents the current ownership, and the capability to unlock $\predi'$ defines the new owner. System time $\systime$ is extracted from the inclusion proof $\pinc$ and used in predicate evaluation by the future verifiers.}\label{fi:unicity-transaction-pred}
    \end{center}
\end{figure}

\end{document}